\documentclass[]{gAPA2e}
\usepackage[dvips]{color} 
\usepackage{amsmath,amsfonts}
\begin{document}
\doi{10.1080/0003681YYxxxxxxxx}
 \issn{1563-504X}
\issnp{0003-6811}
\jvol{00} \jnum{00} \jyear{2009} \jmonth{July}

\markboth{The discrete potential Boussinesq
equation }{Applicable Analysis}


\title{{\itshape 
The discrete potential Boussinesq
equation and its multisoliton solutions}}

\author{Ken-ichi Maruno$^{\rm a}$$^{\ast}$\thanks{$^\ast$Corresponding author. Email: kmaruno@utpa.edu
\vspace{6pt}} and Kenji Kajiwara$^{\rm b}$\\\vspace{6pt}  $^{\rm
a}${\em{Department of Mathematics, The University of Texas-Pan American, 
Edinburg, TX 78539-2999}}; 
$^{\rm b}${\em{Faculty of Mathematics, Kyushu University,
744 Motooka, Fukuoka 819-0395, Japan}}\\
\vspace{6pt}\received{\today} }

\maketitle

\begin{abstract}

An alternate form of
discrete potential Boussinesq equation is 
proposed and its multisoliton solutions are constructed. 
An ultradiscrete potential Boussinesq equation is 
also obtained from the discrete potential 
Boussinesq equation using the ultradiscretization technique. 
The detail of the multisoliton solutions is 
discussed by using the reduction technique. 
The lattice potential Boussinesq equation 
derived by Nijhoff et al. is also 
investigated by using the singularity confinement test. 
The relation between the proposed alternate discrete 
potential Boussinesq equation and 
the lattice potential Boussinesq equation 
by Nijhoff et al. is clarified. 

\begin{classcode}35Q51; 37K10 \end{classcode}\bigskip
\end{abstract}

\section{Introduction}
In this article, we propose 
an alternate form of
discrete potential Boussinesq (BSQ) equation 
\begin{equation}
U_{n-1}^{m+2}U_{n-1}^{m-1}
(-\delta_1U_{n}^{m-1}+\delta_2U_{n}^{m+2})
=U_{n}^{m+1}U_{n-1}^{m}
(-\delta_1U_{n-1}^{m-1}+\delta_2U_{n-1}^{m+2})\,,\label{new-pbsq}
\end{equation}
and 
study the relation to 
the lattice potential BSQ equation proposed by Nijhoff et al.
\cite{Nijhoff,Nijhoff2}
\begin{eqnarray}
&&\frac{p^3-q^3}{p-q+u_{n+1}^{m+1}-u_{n+2}^m}
-\frac{p^3-q^3}{p-q+u_{n}^{m+2}-u_{n+1}^{m+1}}\nonumber\\
&&\quad -u_n^{m+1}u_{n+1}^{m+2}+u_{n+1}^{m}u_{n+2}^{m+1}
+u_{n+2}^{m+2}(p-q+u_{n+1}^{m+2}-u_{n+2}^{m+1})\nonumber\\
&&\quad +u_{n}^{m}(p-q+u_{n}^{m+1}-u_{n+1}^{m})\nonumber\\
&&=(2p+q)(u_{n+1}^m+u_{n+1}^{m+2})
-(p+2q)(u_{n}^{m+1}+u_{n+2}^{m+1})\,.\label{pbsq}
\end{eqnarray}
Here $U=U_n^m$ and $u=u_n^m$ are 
the dynamical field 
variables
at the site $(n,m)$
of a rectangular lattice, and $\delta_1$, $\delta_2$, $p$, $q$ 
are the lattice parameters. 
These are integrable discrete analogues of
the potential BSQ equation
\begin{equation}
3w_{tt}+4c_0w_{xx}-6w_xw_{xx}-w_{xxxx}=0
\label{pbsq-cont}\,,
\end{equation}
which leads to the BSQ equation\cite{bsq-paper,ablowitz}
\begin{equation}
3u_{tt}+4c_0u_{xx}-3(u^2)_{xx}-u_{xxxx}=0\,,
\end{equation}
by $u=w_x$. 
Note that Eq.(\ref{pbsq}) is the lowest order member 
of a hierarchy of lattice equations which is  called the lattice
Gelfand-Dikii(GD) hierarchy\cite{Nijhoff}. See also recent works about 
the lattice potential BSQ equation (\ref{pbsq})\cite{Tongas,Hietarinta}. 

In this paper, we propose 
an alternate form of
discrete potential BSQ equation 
and present multisoliton solutions. 
It is shown that 
the multisoliton solution for the discrete potential BSQ equation
can be constructed from one for 
the Hirota-Miwa (discrete KP) equation using the reduction technique. 
Using the proposed discrete potential BSQ equation, 
we can construct the ultradiscrete potential BSQ equation.
We also study the relation between 
the alternate discrete potential BSQ equation (\ref{new-pbsq}) 
and the lattice potential BSQ equation (\ref{pbsq}). 
Bilinear equations of the lattice potential BSQ equation can be
derived systematically using the singularity confinement (SC) test. 
This reveals the relationship with other discrete potential 
BSQ equations. 

\section{Multisoliton solutions for the Boussinesq equation}

First, we review fundamental results about 
reductions of the Kadomtsev-Petviashvili (KP) equation. 

It is well known that the solutions of the KP equation\cite{kp}
\begin{equation}
(
-
4u_t+6uu_x+u_{xxx})_x+3u_{yy}=0\,,
\end{equation}
can be expressed via a tau function $\tau(x,y,t)$ as
$u(x,y,t) = 2 \frac{\partial^2}{\partial x^2}\log \tau (x,y,t)$
where $\tau(x,y,t)$ satisfies Hirota's bilinear equation
\begin{equation}
(D_x^4-4D_xD_t+3D_y^2)\tau \cdot \tau=0\,,\label{bikp}
\end{equation}
where $D_x, D_t, D_y$ are the Hirota $D$-operators.
Soliton solutions of the bilinear 
equation can be written in terms of the Wronskian
determinant \cite{Freeman,FreemanNimmo,Hirota,satsuma} 
\begin{equation}
\tau(x,y,t)=  \mathop{\rm Wr}(f_1,\cdots , f_N)= 
  \det( f_n^{(n'-1)} )_{1\le n,n'\le N}\,,
\label{e:tauWronskian}
\end{equation}
with $f_n^{(j)}=\partial^j f_n/\partial x^j$, 
and where $f_1,\dots,f_N$ are 
a set of linearly independent solutions of the linear system
\begin{equation}
\frac{\partial f}{\partial y}=
\frac{\partial^2 f}{\partial x^2},\,\quad 
\frac{\partial f}{\partial t}=
\frac{\partial^3 f}{\partial x^3}\,. 
\label{e:linearPDEsyst} 
\end{equation}
For example, ordinary $N$-soliton solutions are obtained by taking
$f_n = e^{\theta_{2n-1}} + e^{\theta_{2n}}$ for $n = 1,\dots,N$,
where
\begin{equation}
\theta_m(x,y,t) =k_mx + k_m^2y +k_m^3t + \theta_{m;0}\,,
\label{e:thetadef}
\end{equation}
for $m=1,\dots,2N$,
where the $4N$ parameters
$k_1<\cdots<k_{2N}$ and $\theta_{1;0},\dots,\theta_{2N;0}$ are real constants. 
For $N=1$, one obtains the single-soliton solution of KP equation:
\begin{equation}
u_{i,j}(x,y,t)=
  \frac{1}{2} (k_j-k_i)^2{\rm sech}^2
\big[\frac{1}{2}(\theta_i-\theta_j)\big]\,, 
\label{e:KPIIsoliton}
\end{equation}
where $i=1$ and $j=2$.
The most general form of the $N$-soliton solution is given by
\begin{equation}
\tau(x,y,t)= \det(A\Theta K)=
  \sum_{1\le m_1<\cdots<m_N\le M}
    V_{m_1,\ldots,m_N}\,A_{m_1,\ldots,m_N}\,
  \exp\,\theta_{m_1,\cdots,m_N}\,,
\label{e:taugeneral}
\end{equation}
where 
$A= (a_{n,m})$ is the $N\times M$ coefficient matrix,
$\Theta={\rm diag}(e^{\theta_1},\cdots ,e^{\theta_M})$, 
$\theta_{m_1,\cdots,m_N}=\theta_{m_1}+\cdots +\theta_{m_N}$,
and the $M\times N$ matrix $K$ is given by $K=(k_m^{n-1})$
\cite{jpa2003,jpa2004,jmp2006,mcs2007,prl2007,jpa2008,Chak-Kodama,boiti}.
$V_{m_1,\dots,m_N}$ is the Vandermonde determinant
$V_{m_1,\dots,m_N}= \prod_{1\le j<j'\le N}(k_{m_{j'}}-k_{m_j})\,$,
and $A_{m_1,\ldots,m_N}$ is the $N\times N$-minor whose $n$-th column 
is respectively given by the $m_n$-th column of the coefficient matrix
for $n = 1, \dots, N$.
The only time dependence in the tau function comes from the
exponential phases $\theta_{m_1,\dots,m_N}$.
Also, for all $G\in \rm{GL}(N,\mathbb{R})$,
the coefficient matrices $A$ and $A'= G\,A$
produce the same solution of the KP equation.
Thus without loss of generality one can consider $A$ to be in 
row-reduced echelon form (RREF).
One can also multiply each column of $A$ by 
an arbitrary positive constant
which can be absorbed in the definition of
$\theta_{1;0},\dots,\theta_{M;0}$.

Real nonsingular (positive) solutions of the KP equation are obtained if 
$k_1<\cdots<k_M$ and all minors of $A$ are nonnegative.
Under these assumptions and some fairly general irreducibility conditions 
on the coefficient matrix,
Eq.(\ref{e:taugeneral}) produces $(N_-,N_+)$-soliton solutions of the 
KP equation with $N_-=M-N$ and $N_+=N$, as in the simpler case of fully 
resonant solutions.
Asymptotic line solitons are given by Eq.(\ref{e:KPIIsoliton})
with the indices $i$ and $j$ labeling the phases $\theta_i$
and $\theta_j$ being swapped 
in the transition between two dominant phase combinations
along the line $\theta_i=\theta_j$.
Asymptotic solitons can thus be uniquely characterized by an index pair 
$[i,j]$
with $1\le i<j\le M$. Recently, line soliton solutions of the KP\,II
equation were classified using this formulation
\cite{jpa2003,jpa2004,jmp2006,mcs2007,jpa2008,Chak-Kodama}.
Elastic 2-soliton solutions 
are classified into three classes:
ordinary (O-type), asymmetric (P-type) and resonant (T-type).
The coefficient matrices corresponding to these classes have 
the following RREFs:
\begin{gather}
A_{\rm O}=\left(\!\begin{array}{cccc}
1 & a & 0 & 0 \\
0 & 0 & 1 & b \\
\end{array}\!\right)\!,\!\quad 
A_{\rm P}=\left(\!\begin{array}{cccc}
1 & 0 & 0 & -b \\
0 & 1 & a & 0 \\
\end{array}\!\right)\!,\!\quad 
A_{\rm T}=\left(\!\begin{array}{cccc}
1 & 0 & -c & -d \\
0 & 1 & a & b \\
\end{array}\!\right)\!,
\label{e:Aelastic2soliton}
\end{gather}
where $a,b,c,d>0$ are free parameters with $ad-bc>0$.
These three types of solutions 
cover disjoint sectors of the 2-soliton parameter space of
amplitudes and directions.
Moreover, their interaction properties are also different.
This difference is obvious in the case of T-type solutions,
but also applies to O-type and P-type solutions, since
P-type solutions only exist for unequal amplitude, and 
the interaction phase shift has the opposite sign for
O-type and P-type solutions.
We remark that 
inelastic 2-soliton solutions fall into four categories, but 
we do not focus on them in this article. 

It is known that the KP equation reduces to 
the Korteweg-de Vries (KdV) 
equation
\begin{equation}
-4u_t+6uu_x+u_{xxx}=0
\end{equation}
by the symmetry constraint 
$\partial u/\partial y=0$\cite{jpa2004,Date-reduction,
JM,Hirota:reduction}. 
In the bilinear form, the bilinear KP equation (\ref{bikp})
is reduced to the bilinear KdV equation
$(D_x^4-4D_xD_t)\tau \cdot \tau =0$
by the constraint of omitting terms including $D_y$, 
which is the so-called 2-reduction.
As mentioned in \cite{jpa2004},
this constraint implies that all the solitons of 
the KdV equation are parallel to the $y$-axis. 
Thus we must have a condition $k_j=-k_i$ for each $[i,j]$-soliton. 
From the ordering $k_1<k_2<\cdots<k_{2N}$, we must assume 
\[
k_1<k_2<\cdots<k_{N}<0\,, \qquad k_{N+j}=-k_{N-j+1}\,\, {\rm for}
\,\, j=1,\cdots, N\,. 
\]
This allows only P-type soliton solution in which $A$-matrix is 
\[
A=
\left(\begin{array}{cccccccc}
  1 & 0 & \cdots & 0 & 0 & \cdots & 0 & a_{1,2N} \\
  0 & 1 & \cdots & 0 & 0 & \cdots & a_{2,2N-1} & 0  \\
\vdots & \vdots & \cdots & \vdots & \vdots & \cdots & \vdots &
 \vdots  \\
 0& 0 & \cdots & 1& a_{N,N} & \cdots & 0 & 0
\end{array}\right)\,.
\]

Let us consider multisoliton solutions of the BSQ equation. 
It is also known that the KP equation is reduced to the BSQ equation
without the term $u_{xx}$, namely
\begin{equation}
3\delta_0u_{t't'}-3(u^2)_{xx}-u_{xxxx}=0\,, \label{bsq-no}
\end{equation}
by the symmetry constraint $\partial u/\partial t=0$
\cite{Date-reduction,JM,Hirota:reduction}. 
Here we introduced 
a new independent variable $t'$ 
such that $y=\sqrt{-\delta_0}t'$, $\delta_0=\pm 1$.
In the bilinear form, the bilinear KP equation (\ref{bikp})
is reduced to 
$(D_x^4-3\delta_0D_{t'}^2)\tau \cdot \tau=0$
by the constraint of omitting terms including $D_t$, 
which is the so-called 3-reduction.
This condition implies that all the solitons of Eq.(\ref{bsq-no}) 
are parallel to the $t$-axis. 
Thus we must have a condition $k_j^2+k_jk_i+k_i^2=0$, i.e.
$k_j=\omega k_i$ or $k_j=\omega^2 k_i$ ($\omega=-1/2+{\rm
i}\sqrt{3}/2$, $\omega^2=-1/2-{\rm i}\sqrt{3}/2$, $\omega^3=1$), 
for $[i,j]$-soliton. 
In the BSQ equation, $k_1,...,k_{2N}$ take complex values in general. 
Since $k_1,...,k_{2N}$ are complex values, we cannot consider the ordering of 
$k_1,...,k_{2N}$ which 
was assumed when we consider KP soliton solutions. 
However, from the constraint of $k_j$,
we have a restriction to the $A$-matrix 
such that each row has only 2 or 3 nonzero elements and each column has 
only one element. 
Note that there are only 2 nonzero elements in each row 
and only one element in each column in the case of the KdV equation 
because of the constraint $k_j=-k_{i}$ \cite{jpa2004}. 
Note also that there is no such 
restriction in the case of the KP equation.  
For 2-soliton solutions, 2 types of 2-soliton interactions are
possible, i.e. 2 elastic soliton interactions (O-type and P-type), 
other 2-soliton interactions are impossible because some columns in 
$A$-matrix have 2 elements. However, there is no distinction 
between O-type and P-type solitons because we do not have the 
ordering of $k_1,...,k_{2N}$. Thus 2-soliton interaction of the BSQ
equation is actually only of one type.  
To get real solutions, we must remove imaginary numbers by the following 
ways.

The case of 2 elements:\\ 
(i) Suppose that we have 2 elements in the $i$-th row of the $A$-matrix. 
Let the corresponding wave numbers of these elements be $k_{j_1}$, $k_{j_2}$; 
(ii) Let $k_{j_2}=\omega k_{j_1}$ (or $k_{j_2}=\omega^2 k_{j_1}$);
(iii) Using the gauge invariance of $\tau$-function, each element in the
Wronskian determinant can be 
$f_i\sim 1+a_{i,j_2}
\exp ((\omega-1)k_{j_1}x+\sqrt{-\delta_0}
(\omega^2-1)k_{j_1}^2t'+\theta_{j_2,0})$ 
(or 
$f_i\sim 1+a_{i,j_2}
\exp ((\omega^2-1)k_{j_1}x+\sqrt{-\delta_0}
(\omega-1)k_{j_1}^2t'+\theta_{j_2,0})$); 
(iv) Reparametrize $k_{j_1}$  and $k_{j_2}$ 
by $\kappa_{i}=(\omega-1)k_{j_1}$ and 
$\Omega_i=\sqrt{-\delta_0}
(\omega^2-1)k_{j_1}^2$ 
(or $\kappa_{i}=(\omega^2-1)k_{j_1}$ and 
$\Omega_i=\sqrt{-\delta_0}
(\omega-1)k_{j_1}^2$). 
Then we have 
$f_i\sim 1+a_{i,j_2}
\exp (\kappa_{i} x+\Omega_{i}t'+\theta_{j_2,0})$
with linear dispersion relations 
$\kappa_{i}^4-3\delta_0\Omega_{i}^2=0$; 
(vi) So a set of $f_i\sim 1+a_{i,j_2}
\exp (\kappa_{i} x+\Omega_{i}t'+\theta_{j_2,0})$ with linear dispersion 
relations $\kappa_{i}^4-3\delta_0\Omega_{i}^2=0$ gives 
the real and regular multisoliton solutions.  

The case of 3 elements:\\ 
(i) Suppose that we have 3 elements in the $i$-th row of the $A$-matrix. 
Let the corresponding wave numbers of these elements be $k_{j_1}$,
$k_{j_2}$, 
$k_{j_3}$; 
(ii) Let $k_{j_2}=\omega k_{j_1}$ and $k_{j_3}=\omega^2 k_{j_1}$;
(iii) Using the gauge invariance of $\tau$-function, each element in the
Wronskian determinant can be 
$f_i\sim 1+a_{i,j_2}
\exp ((\omega-1)k_{j_1}x+(\omega^2-1)k_{j_1}^2y+\theta_{j_2,0})
+a_{i,j_3}
\exp ((\omega^2-1)k_{j_1}x+(\omega-1)k_{j_1}^2y+\theta_{j_3,0})$; 
(iv) Reparametrize parameters $k_{j_1}$, $k_{j_2}$ and $k_{j_3}$ 
by $\kappa_{i,2}=k_{j_2}-k_{j_1}$, 
$\kappa_{i,3}=k_{j_3}-k_{j_1}$, 
$\Omega_{i,2}=\sqrt{-\delta_0}(k_{j_2}^2-k_{j_1}^2)$
and 
$\Omega_{i,3}=\sqrt{-\delta_0}(k_{j_3}^2-k_{j_1}^2)$. 
Then we have 
$f_i\sim 1+a_{i,j_2}
\exp (\kappa_{i} x+\Omega_{i}t'+\theta_{j_2,0})
+a_{i,j_3}
\exp (\kappa_{i} x+\Omega_{i}t'+\theta_{j_3,0})$
with linear dispersion relations 
$\kappa_{i,j}^4-3\delta_0\Omega_{i,j}^2=0$
 for $j=1,2$;
For example, consider $N=1$. 
This gives Y-shape soliton resonance interaction. 

The KP equation is reduced to the BSQ equation
\begin{equation}
3\delta_0u_{t't'}+4c_0u_{xx}-3(u^2)_{xx}-u_{xxxx}=0\,,
\end{equation}
by the symmetry constraint 
$\frac{\partial}{\partial t}=c_0\frac{\partial}{\partial x}$. 
Here we introduced a new independent variable $t'$ 
such that $y=\sqrt{-\delta_0}t'$. 
In the bilinear form, 
the bilinear KP equation (\ref{bikp}) is reduced to 
$(D_x^4-4c_0D_x^2-3\delta_0D_{t'}^2)\tau \cdot \tau=0$
by the constraint of replacing $D_t$ by $c_0D_x$, which 
is the so-called 3-pseudo reduction\cite{Hirota:reduction}.
To realize this constraint in the multisoliton solutions, 
we can add the constraint 
$k_j^3-k_i^3=c_0(k_j-k_i)$,
i.e.
$k_j=\frac{1}{2}\left(-k_i\pm \sqrt{4c_0-3k_i^2}\right)$, 
to the soliton solutions of the KP equation. 
Note that $k_j$ can be real if $4c_0>3k_i^2$.
So we can assume the ordering of $k_1<...<k_{2N}$.
From the constraint we have a restriction to the $A$-matrix 
such that each row has only 2 or 3 nonzero elements and each column has 
only one element. 
For 2-soliton solutions, 2 types of 2-soliton interactions are
possible, i.e. 2 elastic soliton interactions (O-type and P-type), 
other 2-soliton interactions are impossible because some columns in 
$A$-matrix have 2 elements.

Let us consider explicit real and regular multisoliton solutions. 
For the case of $\delta_0=-1$, 
(i) Suppose that we have 2 elements in the $i$-th row of the $A$-matrix. 
Let the corresponding wave numbers of these elements be $k_{j_1}$,
$k_{j_2}$; 
(ii) Let $k_{j_2}=\frac{1}{2}
\left(-k_{j_1}\pm \sqrt{4c_0-3k_{j_1}^2}\right) $;
(iii) Each element in the
Wronskian determinant can be 
$f_i\sim a_{i,j_1}
\exp (k_{j_1}x+k_{j_1}^2t'+\theta_{j_1,0})+
a_{i,j_2}
\exp (k_{j_2}x+k_{j_2}^2t'+\theta_{j_2,0})$. This gives 
multisoliton solutions for the BSQ equation. 

If we have 3 elements in the $i$-th row of the $A$-matrix, 
each element in the
Wronskian determinant can be 
$f_i\sim a_{i,j_1}
\exp (k_{j_1}x+k_{j_1}^2t'+\theta_{j_1,0})+
a_{i,j_2}
\exp (k_{j_2}x+k_{j_2}^2t'+\theta_{j_2,0})
+
a_{i,j_3}
\exp (k_{j_3}x+k_{j_3}^2t'+\theta_{j_3,0})$. 
For $N=1$, this gives Y-shape resonant soliton solution.

For the case of $\delta_0=1$, 
(i) Suppose that we have 2 elements in the $i$-th row of the $A$-matrix. 
Let the corresponding wave numbers of these elements be $k_{j_1}$,
$k_{j_2}$; 
(ii) Let $k_{j_2}=\frac{1}{2}
\left(-k_{j_1}\pm \sqrt{4c_0-3k_{j_1}^2}\right)$;
(iii) Using the gauge invariance of $\tau$-function, each element in the
Wronskian determinant can be 
$f_i\sim 1+a_{i,j_2}
\exp ((k_{j_2}-k_{j_1})x+{\rm i}(k_{j_2}^2-k_{j_1}^2)t'+\theta_{j_2,0})$; 
(iv) Reparametrize $k_{j_1}$ and $k_{j_2}$ by 
$\kappa_{i}=k_{j_2}-k_{j_1}$ and 
$\Omega_i={\rm i}(k_{j_2}^2-k_{j_1}^2)$. 
Then we have 
$f_i\sim 1+a_{i,j_2}
\exp (\kappa_{i} x+\Omega_{i}t'+\theta_{j_2,0})$
with linear dispersion relations 
$\kappa_{i}^4-4c_0 \kappa_{i}^2-3\delta_0\Omega_{i}^2=0$.

Next, we consider the case of complex parameters.  
Suppose that a parameter $k_i$ corresponds to the pivot of $A$-matrix.
If $4c_0<3k_i^2$, 
then we have a complex parameter  
$
k_j=\frac{1}{2}\left(-k_i\pm {\rm i}\sqrt{3k_i^2-4c_0}\right)
$
among $k_1,...,k_{2N}$. 
Since some of $k_1,...,k_{2N}$ 
are complex values, we cannot consider the ordering of 
$k_1,...,k_{2N}$ which 
was assumed when we consider KP soliton solutions. 
However, from the constraint we have a restriction to the $A$-matrix 
such that each row has only 2 or 3 nonzero elements and each column has 
only one element. 
For 2-soliton solutions, 2 types of 2-soliton interactions are
possible, i.e. 2 elastic soliton interactions (O-type and P-type), 
other 2-soliton interactions are impossible because some columns in 
$A$-matrix have 2 elements. However, there is no distinction 
between O-type and P-type solitons because we don't have the 
ordering of $k_1,...,k_{2N}$. Thus 2-soliton interaction of the BSQ
equation is actually only one type again.  
To get real solutions, we must remove imaginary numbers by the following 
ways.

The case of 2 elements:\\ 
(i) Suppose that we have 2 elements in the $i$-th row of the $A$-matrix. 
Let the corresponding wave numbers of these elements be $k_{j_1}$, $k_{j_2}$; 
(ii) Let $k_{j_2}=\frac{1}{2}
\left(-k_{j_1}\pm {\rm i}\sqrt{3k_{j_1}^2-4c_0}\right) $;
(iii) Using the gauge invariance of $\tau$-function, each element in the
Wronskian determinant can be 
$f_i\sim 1+a_{i,j_2}
\exp ((k_{j_2}-k_{j_1})x+(k_{j_2}^2-k_{j_1}^2)y+\theta_{j_2,0})$; 
(iv) Reparametrize $k_{j_1}$ and $k_{j_2}$ 
by $\kappa_{i}=k_{j_2}-k_{j_1}$ and 
$\Omega_i=\sqrt{-\delta_0}(k_{j_2}^2-k_{j_1}^2)$. 
Then we have 
$f_i\sim 1+a_{i,j_2}
\exp (\kappa_{i} x+\Omega_{i}t'+\theta_{j_2,0})$
with linear dispersion relations 
$\kappa_{i}^4-4c_0 \kappa_{i}^2-3\delta_0\Omega_{i}^2=0$; 
(vi) So a set of $f_i\sim 1+a_{i,j_2}
\exp (\kappa_{i} x+\Omega_{i}t'+\theta_{j_2,0})$ with linear dispersion 
relations $\kappa_{i}^4-4c_0 \kappa_{i}^2-3\Omega_{i}^2=0$ gives 
the real and regular multisoliton solutions.  

The case of 3 elements:\\ 
(i) Suppose that we have 3 elements in the $i$-th row of the $A$-matrix. 
Let the corresponding wave numbers of these elements be $k_{j_1}$,
$k_{j_2}$, 
$k_{j_3}$; 
(ii) Let $k_{j_2}=\frac{1}{2}
\left(-k_{j_1}+{\rm i}\sqrt{3k_{j_1}^2-4c_0}\right)$ 
and $k_{j_3}=\frac{1}{2}
\left(-k_{j_1}- {\rm i}\sqrt{3k_{j_1}^2-4c_0}\right)$;
(iii) Using the gauge invariance of $\tau$-function, each element in the
Wronskian determinant can be 
$f_i\sim 1+a_{i,j_2}
\exp ((k_{j_2}-k_{j_1})x+\sqrt{-\delta_0}
(k_{j_2}^2-k_{j_1}^2)t'+\theta_{j_2,0})
+a_{i,j_3}
\exp ((k_{j_3}-k_{j_1})x+\sqrt{-\delta_0}
(k_{j_3}^2-k_{j_1}^2)t'+\theta_{j_3,0})$; 
(iv) Reparametrize $k_{j_1}$, $k_{j_2}$ and $k_{j_3}$ 
by $\kappa_{i,2}=k_{j_2}-k_{j_1}$, 
$\kappa_{i,3}=k_{j_3}-k_{j_1}$, 
$\Omega_{i,2}=\sqrt{-\delta_0}(k_{j_2}^2-k_{j_1}^2)$
and 
$\Omega_{i,3}=\sqrt{-\delta_0}(k_{j_3}^2-k_{j_1}^2)$. 
Then we have 
$f_i\sim 1+a_{i,j_2}
\exp (\kappa_{i} x+\Omega_{i}t'+\theta_{j_2,0})
+a_{i,j_3}
\exp (\kappa_{i} x+\Omega_{i}t'+\theta_{j_3,0})$
with linear dispersion relations 
$\kappa_{i,j}^4-4c_0 \kappa_{i,j}^2-3\delta_0\Omega_{i,j}^2=0$
 for $j=1,2$;
For example, consider $N=1$. 
This gives Y-shape soliton resonance interaction. 

\section{Discrete analogues of the potential Boussinesq equation}
\subsection{The discrete potential Boussinesq equation}

Here, we present the main theorem in this article. 

\begin{theorem}\label{theorem1}
The difference-difference equation
\begin{equation}
U_{n-1}^{m+2}U_{n-1}^{m-1}
(-\delta_1U_{n}^{m-1}+\delta_2U_{n}^{m+2})
=U_{n}^{m+1}U_{n-1}^{m}
(-\delta_1U_{n-1}^{m-1}+\delta_2U_{n-1}^{m+2})\,,\label{new-pbsq2}
\end{equation}
where $\delta_1=a_2(a_1-a_3)$ and $\delta_2=a_3(a_1-a_2)$ and $a_1$,
 $a_2$, $a_3$ are arbitrary real constants,
is an integrable discrete analogue of the potential BSQ equation 
(\ref{pbsq-cont}). 

Moreover, the discrete potential BSQ equation 
has multisoliton solutions
\begin{eqnarray}
&&U_n^m=\frac{\tau_{n+1}^{m}}{\tau_{n}^{m}}\,,
\quad \tau_n^m={\rm det}(\mathcal{A}\Theta P)\,,
\end{eqnarray}
where 
$\mathcal{A}= (\alpha_{n,m})$ is the $N\times 2N$ coefficient matrix,
$\Theta={\rm diag}(e^{\theta_1},\cdots ,e^{\theta_{2N}})$, 
$e^{\theta_{j}}=p_{j}^s(1-p_{j} a_1)^{-n}(1-p_{j} a_2)^{-m}$,
and the 
$2N\times N$ matrix $P$ is given by $P=(p_m^{n-1})$.
where the $2N$ parameters
$p_1,...,p_{2N}$ are real constants. 
The $\mathcal{A}$-matrix has a restriction such that 
each row has only 2 or 3 nonzero elements and each column has 
only one element.
In the case having 2 elements $(i,j_1)$ and $(i,j_2)$ 
in the $i$-th row of the $\mathcal{A}$-matrix, 
$p_{j_2}$ must satisfy a reduction condition
\begin{equation}
p_{j_2}=\frac{1}{a_2}
+\frac{(1-a_3p_{j_1})}{2a_3}\pm 
\frac{\sqrt{a_2(1-a_3p_{j_1})(a_2-4a_3+3a_2a_3p_{j_1})
}}{2a_2a_3}\,.
\end{equation}
\end{theorem}
\begin{proof}
The Hirota-Miwa (discrete KP) equation is written as
\begin{eqnarray}
&&\quad a_1(a_2-a_3)\tau(n_1+1,n_2,n_3) \tau(n_1,n_2+1,n_3+1)\nonumber\\
&&+a_2(a_3-a_1)\tau(n_1,n_2+1,n_3) \tau(n_1+1,n_2,n_3+1)\nonumber\\
&&+a_3(a_1-a_2)\tau(n_1,n_2,n_3+1) \tau(n_1+1,n_2+1,n_3)=0 \,,
\label{eqn:a9}
\end{eqnarray}
where $\tau$ depends on three discrete independent variables $n_1$, $n_2$
and $n_3$, and $a_1$, $a_2$ and $a_3$ are
the difference intervals for $n_1$,$n_2$ and $n_3$, respectively
\cite{Miwa}.

The Casorati determinant solution for the Hirota-Miwa equation
(\ref{eqn:a9}) is as follows\cite{Ohta}:
\begin{eqnarray}
\small
&& \tau (n_1,n_2,n_3)
={\rm det}(\psi_i(n_1,n_2,n_3;s+j-1))_{1\leq i,j \leq N}\,,
\end{eqnarray}
where $\psi_1,...,\psi_N$ are 
a set of linearly independent solutions of the linear system 
\begin{equation}
 \Delta _{n_j}\psi_i(n_1,n_2,n_3;s)=\psi_i(n_1,n_2,n_3;s+1)\,, \quad
(j=1,2,3).\nonumber\label{dispersion}
\end{equation}
Here $\Delta_{n_j}$ are the backward difference operators:
\begin{equation}
  \Delta _{n_j}f(n_j)\equiv \frac{f(n_j)-f(n_j-1)}{a_j}\,, \quad (j=1,2,3).
\end{equation}
For example, ordinary $N$-soliton solutions are obtained by 
taking
\begin{eqnarray}
\psi_i(n_1,n_2,n_3;s)&=&\alpha_{2i-1}
p_{2i-1}^s(1-p_{2i-1} a_1)^{-n_1}(1-p_{2i-1} a_2)^{-n_2}
(1-p_{2i-1} a_3)^{-n_3}\nonumber\\
&\quad &+\alpha_{2i}p_{2i}^s(1-p_{2i} a_1)^{-n_1}(1-p_{2i} a_2)^{-n_2}
(1-p_{2i} a_3)^{-n_3}\,, \label{eqn:b10}
\end{eqnarray}
for $n=1,...,N$ where the $4N$ parameters 
$p_{1}<\cdots<p_{2N}$ and $\alpha_1,\cdots,\alpha_{2N}$ 
are positive real constants. 
The most general form of the $N$-soliton solution is given by
\begin{equation}
\tau(n_1,n_2,n_3)
= \det(\mathcal{A} \Theta P)=
  \sum_{1\le m_1<\cdots<m_N\le 2N}
    V_{m_1,\ldots,m_N}\,\mathcal{A}_{m_1,\ldots,m_N}\,
  \exp\,\theta_{m_1,\cdots,m_N}\,,
\label{e:taugeneral}
\end{equation}
where 
$\mathcal{A}= (\alpha_{n,m})$ is the $N\times 2N$ 
coefficient matrix,
$\Theta={\rm diag}(e^{\theta_1},\cdots ,
{e^{\theta_{2N}}}
)$, 
$e^{\theta_{j}}
=p_{j}^s(1-p_{j} a_1)^{-n_1}(1-p_{j} a_2)^{-n_2}(1-p_{j} a_3)^{-n_3}$,
and the $2N\times N$ matrix $P$ is given by $P=(p_m^{n-1})$.
$V_{m_1,\dots,m_N}$ is the Vandermonde determinant
$V_{m_1,\dots,m_N}= \prod_{1\le j<j'\le N}(p_{m_{j'}}-p_{m_j})\,$,
and $\mathcal{A}_{m_1,\ldots,m_N}$ is the $N\times N$-minor whose $n$-th column 
is respectively given by the $m_n$-th column of the coefficient matrix
for $n = 1, \dots, N$.
For all $G\in \rm{GL}(N,\mathbb{R})$,
the coefficient matrices $\mathcal{A}$ 
and $\mathcal{A}'= G\,\mathcal{A}$
produce the same solution of the Hirota-Miwa equation.
Thus without loss of generality one can consider $\mathcal{A}$ to be in 
RREF.

Let us consider the 3-reduction condition
\begin{equation}
\tau(n_1,n_2,n_3)\Bumpeq \tau(n_1,n_2-2,n_3-1)\,,
\end{equation}
for arbitrary $n_1,n_2$ and $n_3$
Applying this reduction condition, we can omit the 
dependence on $n_3$ 
and obtain the bilinear form 
\begin{eqnarray}
&&a_1(a_2-a_3)\tau(n_1+1,n_2) \tau(n_1,n_2-1)
+a_2(a_3-a_1)\tau(n_1,n_2+1) \tau(n_1+1,n_2-2)\nonumber\\
&&\quad +a_3(a_1-a_2)\tau(n_1,n_2-2) \tau(n_1+1,n_2+1)=0\,.
\end{eqnarray}
After the change of variables 
$n_1\to n$, $n_2\to m$, $\tau(n_1,n_2)\to \tau_n^m$, 
we obtain
\begin{eqnarray}
&&a_1(a_2-a_3)\tau_{n+1}^{m+1}\tau_n^{m}
+a_2(a_3-a_1)\tau_n^{m+2}\tau_{n+1}^{m-1}
+a_3(a_1-a_2)\tau_n^{m-1}\tau_{n+1}^{m+2}=0\,,
\nonumber\\
\end{eqnarray}
which is the bilinear form of the discrete potential BSQ equation. 

Now we impose a constraint on the parameters of the solution
so that the reduction condition is satisfied.
For simplicity, we consider the case in which $\psi_1,...,\psi_{N}$ 
have 2 terms. Then we observe
\small
\begin{eqnarray}
&&\psi_i(n_1,n_2+2,n_3+1;s)\nonumber\\
  &&=p_{j_1}^s(1-p_{j_1} a_1)^{-n_1}
(1-p_{j_1} a_2)^{-n_2-2}(1-p_{j_1} a_3)^{-n_3-1}\nonumber\\
&&\quad +p_{j_2}^s(1-p_{j_2} a_1)^{-n_1}(1-p_{j_2} a_2)^{-n_2-2}
(1-p_{j_2}a_3)^{-n_3-1}\nonumber\\
  &&=p_{j_1}^s(1-p_{j_1} a_1)^{-n_1}(1-p_{j_1} a_2)^{-n_2-2}
(1-p_{j_1}a_3)^{-n_3-1}
\nonumber\\
&&
\times\left[ 1+C_i
\left(\frac{p_{j_2}}{p_{j_1}}\right)^s
\left(\frac{1-p_{j_2}a_1}{1-p_{j_1}a_1}\right)^{-n_1}
\left(\frac{1-p_{j_2}a_2}{1-p_{j_1}a_2}\right)^{-n_2}
\left(\frac{1-p_{j_2}a_3}{1-p_{j_1}a_3}\right)^{-n_3}
\right]\,,\nonumber\\
\end{eqnarray}
where 
\[
C_i=
\left(\frac{1-p_{j_2}a_2}{1-p_{j_1}a_2}\right)^{-2}
\left(\frac{1-p_{j_2}a_3}{1-p_{j_1}a_3}\right)^{-1} \,.
\]
\normalsize
If we apply the reduction condition
\begin{equation}
(1-p_{j_1}a_2)^2(1-p_{j_1}a_3)=(1-p_{j_2}a_2)^2(1-p_{j_2}a_3)\,,
\end{equation}
i.e. 
\begin{equation}
p_{j_2}=\frac{1}{a_2}
+\frac{(1-a_3p_{j_1})}{2a_3}\pm 
\frac{\sqrt{a_2(1-a_3p_{j_1})(a_2-4a_3+3a_2a_3p_{j_1})
}}{2a_2a_3}\,,
\end{equation}
we obtain
\begin{equation}
\psi_i(n_1,n_2,n_3;s)=(1-p_{j_2}a_2)^{-2}(1-p_{j_2}a_3)^{-1}
\psi_i(n_1,n_2-2,n_3-1;s)\,,
\end{equation}
Finally we have a reduction condition
\begin{eqnarray}
\tau(n_1,n_2,n_3)
=\prod_{k=1}^N(1-p_ka_2)^{-2}(1-p_k a_3)^{-1}
\tau(n_1,n_2-2,n_3-1)\,.\nonumber\\
\end{eqnarray}

In a similar way discussed in section 2, 
we can consider the general multisoliton solutions of 
the discrete potential BSQ equation. 
\end{proof}

The reality condition is 
$
p_{j_1}<1/a_3\,,\, p_{j_1}>(-a_2+4a_3)/(3a_2a_3) 
$
for $a_2<a_3$, and 
$
p_{j_1}>1/a_3\,,\, p_{j_1}<(-a_2+4a_3)/(3a_2a_3)
$
for $a_2>a_3$. 
With parameters satisfying these conditions, we can 
construct real and regular multisoliton solutions using the 
formula in Theorem \ref{theorem1}. 

Otherwise, we must use the technique of reparametrization which 
was used in section 2. The procedure is as follows. 
(i) Suppose that we have 2 elements in the $i$-th row of the $A$-matrix. 
Let the corresponding complex wave numbers of these elements be $p_{j_1}$ and 
$p_{j_2}$; 
(iii) Using the gauge invariance of $\tau$-function, each element in the
Wronskian determinant can be 
$f_i\sim 1+\alpha_{i,j_2}
K_i^{-n}\Omega_i^{-m}$ where 
$K_i=(1-a_1p_{j_2})/(1-a_1p_{j_1})$, 
$\Omega_i=(1-a_2p_{j_2})/(1-a_2p_{j_1})$; 
(iv) Choose new parameters $K_i$ and $\Omega_i$ to be real numbers. 

It is easy to take the ultradiscrete limit in Eq.(\ref{new-pbsq2}). 
\begin{theorem}
The ultradiscrete potential BSQ equation is 
\begin{eqnarray}
&&V_{n-1}^{m+2}+V_{n-1}^{m-1}+L_n^{m+2}
=V_{n}^{m+1}+V_{n-1}^{m}+L_{n-1}^{m+2}\,,\\
&&
V_n^{m+2}=\max (L_n^{m+2},V_n^{m-1}+c_1)-c_2\,.
\end{eqnarray} 
\end{theorem}
\begin{proof}
Use the standard procedure of ultradiscretization. 
Equation (\ref{new-pbsq2}) is rewritten in the form of 
\begin{eqnarray}
&&U_{n-1}^{m+2}U_{n-1}^{m-1}
I_n^{m+2}
=U_{n}^{m+1}U_{n-1}^{m}
I_{n-1}^{m+2}\,,\\
&&I_n^{m+2}=-\delta_1U_{n}^{m-1}+\delta_2U_{n}^{m+2}\,.
\end{eqnarray}
Introduce new variables $U_n^m=\exp(V_n^m/\epsilon)$, 
$I_n^m=\exp(L_n^m/\epsilon)$, $\delta_1=\exp(c_1/\epsilon)$, 
$\delta_2=\exp(c_2/\epsilon)$. 
Then take the limit 
$\epsilon\to 0^+$ using the formula 
$\lim_{\epsilon \to 0^+}\epsilon\ln(\exp(A/\epsilon)+\exp(B/\epsilon))
=\max(A,B)$ for $A, B \in \mathbb{R}_{\geq 0}$\cite{tokihiro}. 
\end{proof} 
{\bf Remark 1:}\\
Applying the gauge transformation $U_n^m\to (-1)^{n+m}U_n^m$ 
to Eq.(\ref{new-pbsq2}), we have 
\begin{eqnarray}
&&U_{n-1}^{m+2}U_{n-1}^{m-1}
I_n^{m+2}
=U_{n}^{m+1}U_{n-1}^{m}
I_{n-1}^{m+2}\,,\\
&&I_n^{m+2}=\delta_1U_{n}^{m-1}+\delta_2U_{n}^{m+2}\,.
\label{new-pbsq3}
\end{eqnarray}
This form gives a more symmetric form of 
the ultradiscrete BSQ equation
\begin{eqnarray}
&&V_{n-1}^{m+2}+V_{n-1}^{m-1}+L_n^{m+2}
 =V_{n}^{m+1}+V_{n-1}^{m}+L_{n-1}^{m+2}\,,\\
&&L_n^{m+2}=\max(V_n^{m-1}+c_1,V_n^{m+2}+c_2)\,.
\end{eqnarray} 

{\bf Remark 2:}\\
Date et al. proposed another discrete potential BSQ equation\cite{Date-bsq}
\begin{eqnarray}
&&v_{n-1}^{m}v_{n}^{m-1}
(a_{1}(a_2-a_3)v_{n+1}^{m}+a_{2}(a_3-a_1)v_{n}^{m+1})\nonumber\\
&&\quad =v_{n-1}^{m-1}v_{n+1}^{m+1}
(a_{1}(a_2-a_3)v_{n}^{m-1}+a_{2}(a_3-a_1)v_{n-1}^{m})\,.\label{date-pbsq}
\end{eqnarray}
By the transformation $v_n^m=\tau_{n+1}^{m+1}/\tau_n^m$, 
we obtain a bilinear equation
\begin{eqnarray}
&&a_1(a_2-a_3)\tau_{n+1}^{m}\tau_{n-1}^{m}
+a_2(a_3-a_1)\tau_n^{m+1}\tau_{n}^{m-1}
+a_3(a_1-a_2)\tau_{n-1}^{m-1}\tau_{n+1}^{m+1}=0\,.
\nonumber\\\label{date-bilinear}
\end{eqnarray}
This bilinear equation is obtained by adding the reduction condition
\begin{equation}
\tau(n_1,n_2,n_3)\Bumpeq \, \tau(n_1-1,n_2-1,n_3-1)\,,
\end{equation}
which gives yet another 3-reduction. 

For this discrete potential BSQ equation, we can also make an 
ultradiscrete analogue of the potential BSQ equation
\begin{eqnarray}
&&X_{n-1}^m+X_n^{m-1}+Y_n^{m+1}=X_{n-1}^{m-1}+X_{n+1}^{m+1}+Y_{n-1}^m\,,\\
&&Y_n^m=\max(X_{n+1}^{m-1}+c_1,X_n^m+c_2)\,,
\end{eqnarray}
taking the ultradiscrete limit after setting 
$v_n^m=\exp(X_n^m/\epsilon)$, $w_n^m=\exp(Y_n^m/\epsilon)$, 
$\alpha_1=\exp(c_1/\epsilon)$, $\alpha_2=\exp(c_2/\epsilon)$ where 
$w_n^m=\alpha_1v_{n+1}^{m-1}+\alpha_2 v_n^m$, $\alpha_1=a_1(a_2-a_3)$, 
$\alpha_2=a_2(a_3-a_1)$.

\subsection{The lattice potential Boussinesq equation}

Singularity confinement (SC) test was proposed by Grammaticos et al. 
as a detector of integrability in discrete systems\cite{GRP:SC}.
This property has been applied to several problems
\cite{RGH:dP,PGR:accel,RGS:dP,MKNO,MKO}.
The SC test is also powerful tool for constructing solutions for discrete 
integrable systems\cite{RGS:dP,MKNO,MKO}. In this section, 
we apply the SC test to the lattice potential BSQ equation 
of Nijhoff et al. and 
obtain bilinear equations using the result of SC test.

We start from the slightly simplified form of the 
lattice potential BSQ equation
\begin{eqnarray}
&&\frac{p^3-q^3}{p-q+u_{n+1}^{m+1}-u_{n+2}^m}
-\frac{p^3-q^3}{p-q+u_{n}^{m+2}-u_{n+1}^{m+1}}\nonumber\\
&&\quad =
(p+2q+
u_{n+1}^{m}-u_{n+2}^{m+2})
(p-q+u_{n+1}^{m+2}-u_{n+2}^{m+1})\nonumber\\
&&\qquad -(p+2q+u_{n}^m-u_{n+1}^{m+2})
(p-q+u_{n}^{m+1}-u_{n+1}^{m})\label{pbsq1} \,.
\end{eqnarray}

Introducing new variables
\begin{eqnarray}
&&I_n^m=p-q+u_{n-1}^{m+1}-u_n^m \,,
\quad V_n^m=p+2q+u_{n-1}^{m-2}-u_n^m \,,\label{ivtrans}
\end{eqnarray}
Eq.(\ref{pbsq1}) is written as 
\begin{eqnarray}
&&\frac{p^3-q^3}{I_{n+2}^m}
-\frac{p^3-q^3}{I_{n+1}^{m+1}}
=V_{n+2}^{m+2}I_{n+2}^{m+1}
-V_{n+1}^{m+2}I_{n+1}^m \,,\label{ivpbsq1}\\
&&I_n^m+V_{n+1}^{m+2}=I_{n+1}^{m+2}+V_n^{m+3}\,.\label{ivpbsq2}
\end{eqnarray}
After the transformation of the independent variables 
$m+n\to m$,
Eqs.(\ref{pbsq1}),(\ref{ivtrans}) and (\ref{ivpbsq1}), (\ref{ivpbsq2}) 
are rewritten in 
the following form: 
\begin{eqnarray}
&&\frac{p^3-q^3}{p-q+u_{n+1}^{m+2}-u_{n+2}^{m+2}}
-\frac{p^3-q^3}{p-q+u_{n}^{m+2}-u_{n+1}^{m+2}}\nonumber\\
&&\quad =
(p+2q+
u_{n+1}^{m+1}-u_{n+2}^{m+4})
(p-q+u_{n+1}^{m+3}-u_{n+2}^{m+3})\nonumber\\
&&\quad -(p+2q+u_{n}^m-u_{n+1}^{m+3})
(p-q+u_{n}^{m+1}-u_{n+1}^{m+1}) \,,
\end{eqnarray}
\begin{equation}
I_n^m=p-q+u_{n-1}^{m}-u_n^m,\quad V_n^m=p+2q+u_{n-1}^{m-3}-u_n^m \,,
\label{IV-trans}
\end{equation}
and 
\begin{eqnarray}
&&\frac{p^3-q^3}{I_{n+2}^{m+2}}
-\frac{p^3-q^3}{I_{n+1}^{m+2}}
=
V_{n+2}^{m+4}I_{n+2}^{m+3}
-V_{n+1}^{m+3}I_{n+1}^{m+1} \,,\\
&&I_n^m+V_{n+1}^{m+3}=I_{n+1}^{m+3}+V_n^{m+3}\,.
\end{eqnarray}
After some calculation,
we get the form which is suitable to perform the SC test:
\begin{eqnarray}
&&V_n^m=\frac{1}{I_n^{m-1}}\left(V_{n-1}^{m-1}I_{n-1}^{m-3}
+\frac{p^3-q^3}{I_n^{m-2}}-\frac{p^3-q^3}{I_{n-1}^{m-2}}\right) \,,
\label{scform1}\\
&&I_n^m=I_{n-1}^{m-3}+V_n^m-V_{n-1}^m \,.\label{scform2}
\end{eqnarray}

Now we perform the SC test. 
Suppose that we have some initial data which 
can be evolved using the above system.  
Let us assume that during the successive 
applications of the above system the value of 
$I$ at $(m-1,n)$ becomes zero. 
The point where this occurs depends on the initial data. 
Then we have the following pattern of $0$ and $\infty$: 
\begin{eqnarray*}
&&\{I_n^{m-1},I_n^{m},I_{n+1}^{m},I_{n+1}^{m+1}\}
\to \{0,\infty,\infty,0\}\,,
\quad \{V_n^m,V_{n+1}^{m+3}\}\to \{\infty,\infty \}\,.
\end{eqnarray*}
One can see nonzero finite values for all dependent variables 
in further steps, 
so the singularity is perfectly confined. 
Suppose that this pattern was created by a function $F_n^m$ 
which has a zero at $(m,n)$. \\

Using the above singularity
pattern, we obtain the independent variable transformation 
\begin{eqnarray}
I_n^m=\alpha \frac{F_n^{m+1} F_{n-1}^{m-1}}{F_n^{m}F_{n-1}^{m}}\,.
\label{I-transform}
\end{eqnarray}
Since $u_n^m$ and $I_n^m$ are related by Eq.(\ref{IV-trans}), 
$u_n^m$ should have $F_n^m$ in the denominator. 
Thus we set 
\begin{equation}
u_n^m=\beta\frac{G_n^{m}}{F_n^{m}}\,.\label{u-transform}
\end{equation}  
From Eqs.(\ref{ivtrans}) and (\ref{u-transform}), 
we obtain
\begin{eqnarray}
I_n^m=p-q+\beta \frac{G_{n-1}^{m}}{F_{n-1}^{m}}
-\beta \frac{G_n^m}{F_n^m}\, ,\quad
V_n^m=p+2q+\beta\frac{G_{n-1}^{m-3}}{F_{n-1}^{m-3}}
-\beta \frac{G_n^m}{F_n^m}\,.\label{pbsqtrans2}
\end{eqnarray}
Using Eqs.(\ref{scform1}), (\ref{scform2}), (\ref{I-transform})
and (\ref{pbsqtrans2}), we obtain
the following equations:
\begin{eqnarray}
&&(p-q)F_{n-1}^{m-1}F_n^{m-1}
+\beta G_{n-1}^{m-1}F_n^{m-1}
-\beta F_{n-1}^{m-1}G_n^{m-1}
=\alpha F_n^m F_{n-1}^{m-2}\, ,\label{bisa1}\\
&&\frac{p^3-q^3}{\alpha}F_n^{m+1}F_{n+1}^{m}
-\alpha (p+2q)F_{n+1}^{m+2}F_n^{m-1}\nonumber\\
&&\quad -\alpha \beta G_n^{m-1}F_{n+1}^{m+2}
+\alpha \beta F_n^{m-1}G_{n+1}^{m+2}
= \gamma F_n^mF_{n+1}^{m+1}\, ,\label{bisa2}\\
&&\alpha \frac{F_n^m F_{n-1}^{m-2}}
{F_n^{m-1}F_{n-1}^{m-1}}
-\frac{p^3-q^3}{\alpha^2}
\frac{F_{n-1}^{m+1}F_{n}^{m}}
{F_{n}^{m+2}F_{n-1}^{m-1}}
+\frac{\gamma}{\alpha}
\frac{F_{n-1}^{m}F_{n}^{m+1}}
{F_{n}^{m+2}F_{n-1}^{m-1}}\nonumber\\
&&=
\alpha \frac{F_{n+1}^{m+3} F_{n}^{m+1}}
{F_{n+1}^{m+2}F_{n}^{m+2}}
-\frac{p^3-q^3}{\alpha^2}
\frac{F_{n}^{m+1}F_{n+1}^{m}}
{F_{n+1}^{m+2}F_{n}^{m-1}}
+\frac{\gamma}{\alpha}
\frac{F_{n}^{m}F_{n+1}^{m+1}}
{F_{n+1}^{m+2}F_{n}^{m-1}}\,,\label{sabun}
\end{eqnarray}
where $\gamma$ is a decoupling constant. 
Note that $V$ is written in the following form:
\begin{equation}
V_n^m=(p+2q)+\beta \frac{G_{n-1}^{m-2}}
{F_{n-1}^{m-2}}
-\beta  \frac{G_n^{m+1}}{F_n^{m+1}}
=\frac{p^3-q^3}{\alpha^2}
\frac{F_{n-1}^{m}F_{n}^{m-1}}
{F_{n}^{m+1}F_{n-1}^{m-2}}
-\frac{\gamma}{\alpha}
\frac{F_{n-1}^{m-1}F_{n}^{m}}
{F_{n}^{m+1}F_{n-1}^{m-2}}\,.
\end{equation}
Assuming $\gamma=0$, we obtain
\begin{eqnarray}
&&(p-q)F_{n-1}^{m-1}F_n^{m-1}
+\beta G_{n-1}^{m-1}F_n^{m-1}
-\beta F_{n-1}^{m-1}G_n^{m-1}
=\alpha F_n^m F_{n-1}^{m-2}\,,\\
&&\frac{p^3-q^3}{\alpha}F_n^{m+1}F_{n+1}^{m}
-\alpha (p+2q)F_{n+1}^{m+2}F_n^{m-1}
-\alpha \beta G_n^{m-1}F_{n+1}^{m+2}\nonumber\\
&&\qquad +\alpha \beta F_n^{m-1}G_{n+1}^{m+2}
=0 \,,\\
&&
\alpha F_{n}^{m-1} F_{n+1}^{m+3}-\frac{p^3-q^3}{\alpha^2}
F_{n}^{m+2} F_{n+1}^{m}
=\delta F_{n}^mF_{n+1}^{m+2}\,,
\end{eqnarray}
where $\delta$ is a decoupling constant. 
After changing back to original independent variables ($m-n\to m$), 
we have 
\begin{eqnarray}
&&(p-q)F_{n-1}^{m}F_n^{m-1}
+\beta G_{n-1}^{m}F_n^{m-1}
-\beta F_{n-1}^{m}G_n^{m-1}
=\alpha F_n^m F_{n-1}^{m-1}\,,\label{pbsqbi0}\\
&&\frac{p^3-q^3}{\alpha}F_n^{m+1}F_{n+1}^{m-1}
-\alpha (p+2q)F_{n+1}^{m+1}F_n^{m-1}
-\alpha \beta G_n^{m-1}F_{n+1}^{m+1}\nonumber\\
&&\qquad +\alpha \beta F_n^{m-1}G_{n+1}^{m+1}
=0 \,,\label{pbsqbi1}\\
&&
\alpha F_{n}^{m-1} F_{n+1}^{m+2}-\frac{p^3-q^3}{\alpha^2}
F_{n}^{m+2} F_{n+1}^{m-1}
=\delta F_{n}^mF_{n+1}^{m+1}\,.\label{pbsqbi2}
\end{eqnarray}
Note that Eq.(\ref{pbsqbi2}) can be 
derived by eliminating $G$ from Eqs.(\ref{pbsqbi0})
and (\ref{pbsqbi1}).
This is a discrete analogue of a bilinear form of the potential
BSQ equation (\ref{pbsq-cont}). 
Thus we have the following theorem. 
\begin{theorem}
Solutions of Eq.(\ref{pbsq1}), i.e. Eqs.(\ref{ivpbsq1}) and (\ref{ivpbsq2}), 
are expressed in the following form:
 \begin{eqnarray*}
&&u_n^m=\frac{G_{n}^{m}}{F_{n}^{m}}\,,\quad 
I_n^m=\frac{F_{n}^{m+1}F_{n-1}^{m-1}}{F_{n}^{m}F_{n-1}^m}
=p-q+\frac{G_{n-1}^{m+1}}{F_{n-1}^{m+1}}
-\frac{G_n^m}{F_n^m}\,,\\
&&V_n^m=p+2q+\frac{G_{n-1}^{m-2}}{F_{n-1}^{m-2}}
-\frac{G_n^m}{F_n^m}\,, 
\end{eqnarray*}
where $F_n^m$ and $G_n^m$ satisfy 
Eqs.(\ref{pbsqbi0}), (\ref{pbsqbi1}) and (\ref{pbsqbi2}). 
Moreover, $F_n^m$ is given by $\tau_n^m$ of the discrete potential 
BSQ equation (\ref{new-pbsq2}) in Theorem \ref{theorem1} 
when parameters have relations
\[
\frac{\alpha}{\delta}=\frac{a_3(a_1-a_2)}{a_1(a_3-a_2)}\,,\quad 
\frac{p^3-q^3}{\alpha^2\delta}= \frac{a_2(a_1-a_3)}{a_1(a_3-a_2)}\,.
\]
\end{theorem}
Furthermore, we have the following theorem. 
\begin{theorem}
The above $\tau$-functions $F_n^m$ and $G_n^m$ are 
given by 
\begin{equation}
F_n^m=\tau_n^m, \quad G_n^m=\rho_n^m+
\left[
\left(\frac{p}{\beta}-\frac{1}{a_1}\right)n
+\left(\frac{q}{\beta}-\frac{1}{a_2}\right)m
\right]
\tau_n^m\,.
\end{equation}
$\tau_n^m$ is defined by Eq.(\ref{tau-def}) 
($\tau_n^m$ is also given in Theorem 3.1), 
$\rho_n^m$ is defined by Eq.(\ref{rho-def}).
Here parameters should be chosen as 
$a_3=-a_2/2$, 
\[
\alpha=\left(\frac{1}{a_1}-\frac{1}{a_2}\right)\beta\,,
\quad \delta=-\frac{3\beta}{a_2}\,,\quad 
\beta=\left[
\frac{(p^3-q^3)a_1^3a_2^3}{(2a_1+a_2)(a_1-a_2)^2}\right]^{\frac{1}{3}}\,. 
\]
\end{theorem}
\begin{proof}
Let us introduce the following $\tau$-functions using 
the Freeman-Nimmo notation \cite{FreemanNimmo,Ohta,MKO,OKMS}:
{\small
\begin{eqnarray}
&&\tau_n^m(k)=
\left|
\begin{array}{cccc}
\psi_1(n,m,k;s) & \psi_1(n,m,k;s+1) & \cdots & 
\psi_1(n,m,k;s+N-1)\\
\psi_2(n,m,k;s) & \psi_2(n,m,k;s+1) & \cdots & 
\psi_2(n,m,k;s+N-1)\\
\vdots & \vdots &\cdots &\vdots\\
\psi_N(n,m,k;s) & \psi_N(n,m,k;s+1) & \cdots & 
\psi_N(n,m,k;s+N-1)
\end{array}
\right|\nonumber\\
&&\quad \qquad =|\mathbf{0},\mathbf{1},\cdots,
 \mathbf{N-2},\mathbf{N-1}|\,,\label{tau-def}\\
&&\rho_n^m(k)=
\left|
\begin{array}{cccc}
\psi_1(n,m,k;s) & \cdots & \psi_1(n,m,k;s+N-2) & 
\psi_1(n,m,k;s+N)\\
\psi_2(n,m,k;s) & \cdots & \psi_2(n,m,k;s+N-2) &  
\psi_2(n,m,k;s+N)\\
\vdots & \cdots &\vdots &\vdots\\
\psi_N(n,m,k;s) & \cdots & \psi_N(n,m,k;s+N-2) & 
\psi_N(n,m,k;s+N)
\end{array}
\right|\nonumber\\
&&\quad \qquad =|\mathbf{0},\mathbf{1},\cdots,
 \mathbf{N-2},\mathbf{N}|
\,,\label{rho-def}
\end{eqnarray}
}
where $\psi_1,...,\psi_N$ are 
a set of linearly independent solutions of the linear system 
\begin{equation}
\Delta _{n}\psi_i(n,m,k;s)=\Delta _{m}\psi_i(n,m,k;s)
=\Delta _{k}\psi_i(n,m,k;s)=\psi_i(n,m,k;s+1)\,.\nonumber
\end{equation}
Here $\Delta_{n}, \Delta_{m}, \Delta_{k}$ 
are the backward difference operators
\begin{eqnarray*}
&&  \Delta _{n}f(n)\equiv \frac{f(n)-f(n-1)}{a_1}\,,
\quad \Delta _{m}f(m)\equiv \frac{f(m)-f(m-1)}{a_2}\,,
\quad \Delta _{k}f(k)\equiv \frac{f(k)-f(k-1)}{a_3}\,.
\end{eqnarray*}
Then we notice the following difference formulas \cite{Ohta,OKMS}
\begin{eqnarray*}
&&a_1\tau_{n+1}^m(k)=|\mathbf{0},\mathbf{1},\cdots, \mathbf{N-3}, 
\mathbf{N-2},\mathbf{N-2_{n+1}}|\,,\\
&&a_2\tau_{n}^{m+1}(k)=|\mathbf{0},\mathbf{1},\cdots, \mathbf{N-3}, 
\mathbf{N-2},\mathbf{N-2_{m+1}}|\,,\\
&&a_3\tau_{n}^{m}(k+1)=|\mathbf{0},\mathbf{1},\cdots, \mathbf{N-3}, 
\mathbf{N-2},\mathbf{N-2_{k+1}}|\,,\\
&&(a_1-a_2)\tau_{n+1}^{m+1}(k)=
|\mathbf{0},\mathbf{1},\cdots, \mathbf{N-3}, 
\mathbf{N-2}_{m+1},\mathbf{N-2_{n+1}}|\,,\\
&&(a_2-a_3)\tau_{n+1}^{m}(k+1)=
|\mathbf{0},\mathbf{1},\cdots, \mathbf{N-3}, 
\mathbf{N-2}_{k+1},\mathbf{N-2_{m+1}}|\,,\\
&&a_1\rho_{n+1}^m(k)-\tau_{n+1}^{m}(k)=
|\mathbf{0},\mathbf{1},\cdots, \mathbf{N-3}, 
\mathbf{N-1},\mathbf{N-2_{n+1}}|\,,\\
&&a_2\rho_{n}^{m+1}(k)-\tau_{n}^{m+1}(k)=
|\mathbf{0},\mathbf{1},\cdots, \mathbf{N-3}, 
\mathbf{N-1},\mathbf{N-2_{m+1}}|\,,\\
&&a_3\rho_{n}^{m}(k+1)-\tau_{n}^{m}(k+1)=
|\mathbf{0},\mathbf{1},\cdots, \mathbf{N-3}, 
\mathbf{N-1},\mathbf{N-2_{k+1}}|\,.
\end{eqnarray*}
From Pl\"ucker relations and these difference formulas, 
we obtain the following three bilinear equations
\begin{eqnarray}
&&\rho_n^{m+1}(k)\tau_{n+1}^m(k)
-\rho_{n+1}^{m}(k)\tau_{n}^{m+1}(k)\nonumber\\
&&\qquad =\left(\frac{1}{a_1}-\frac{1}{a_2}\right)
(\tau_{n+1}^{m+1}(k)\tau_n^m(k)-\tau_n^{m+1}(k)\tau_{n+1}^m(k))\,,
\\
&&\rho_n^{m}(k+1)\tau_{n+1}^m(k)
-\rho_{n+1}^{m}(k)\tau_{n}^{m}(k+1)\nonumber\\
&&\qquad 
=\left(\frac{1}{a_1}-\frac{1}{a_3}\right)
(\tau_{n+1}^{m}(k+1)\tau_n^m(k)-\tau_{n+1}^{m}(k)\tau_{n}^m(k+1))\,,
\\
&&a_1(a_2-a_3)\tau_{n+1}^m(k)\tau_{n}^{m+1}(k+1)
+a_2(a_3-a_1)\tau_{n}^{m+1}(k)\tau_{n+1}^{m}(k+1)\nonumber\\
&&\qquad 
+a_3(a_1-a_2)\tau_{n}^m(k+1)\tau_{n+1}^{m+1}(k)=0\,.
\end{eqnarray}
Applying the 3-reduction condition $\tau_n^m(k+1)
\Bumpeq
 \tau_n^{m-2}(k)$ and setting $f_n^m=\tau_n^m(k)$, 
$g_n^m=\rho_n^m(k)$, 
we obtain 
\begin{eqnarray}
&&g_n^{m+1}f_{n+1}^m
-g_{n+1}^{m}f_{n}^{m+1}
=\left(\frac{1}{a_1}-\frac{1}{a_2}\right)
(f_{n+1}^{m+1}f_n^m-f_n^{m+1}f_{n+1}^m)\,,\label{red-bsqbilinear0}
\\
&&g_n^{m-2}f_{n+1}^m
-g_{n+1}^{m}f_{n}^{m-2}
=\left(\frac{1}{a_1}-\frac{1}{a_3}\right)
(f_{n+1}^{m-2}f_n^m-f_{n+1}^{m}f_{n}^{m-2})\,,\label{red-bsqbilinear1}
\\
&&a_1(a_2-a_3)f_{n+1}^mf_{n}^{m-1}
+a_2(a_3-a_1)f_{n}^{m+1}f_{n+1}^{m-2}
+a_3(a_1-a_2)f_{n}^{m-2}f_{n+1}^{m+1}=0\,.\nonumber\\
\label{red-bsqbilinear2}
\end{eqnarray}
To compare these bilinear equations with bilinear equations 
(\ref{pbsqbi0})-(\ref{pbsqbi2}), 
we set 
\[
g_n^m=h_n^m+
\left(\frac{n}{a_1}+\frac{m}{a_2}\right)f_n^m\,,
\quad a_3=-\frac{a_2}{2}\,,\quad G_n^m=H_n^m+\frac{pn+qm}{\beta}F_n^m\,. 
\]
Then bilinear equations 
(\ref{red-bsqbilinear0})-(\ref{red-bsqbilinear2}) and 
(\ref{pbsqbi0})-(\ref{pbsqbi2}) 
are rewritten as 
\begin{eqnarray}
&&h_n^{m+1}f_{n+1}^m
-h_{n+1}^{m}f_{n}^{m+1}
=\left(\frac{1}{a_1}-\frac{1}{a_2}\right)
f_{n+1}^{m+1}f_n^m\,,\label{new-red-bsqbilinear0}
\\
&&h_n^{m-2}f_{n+1}^m
-h_{n+1}^{m}f_{n}^{m-2}
=\left(\frac{1}{a_1}+\frac{2}{a_2}\right)
f_{n+1}^{m-2}f_n^m\,,\label{new-red-bsqbilinear1}
\\
&&3a_1f_{n+1}^mf_{n}^{m-1}
-(2a_1+a_2)f_{n}^{m+1}f_{n+1}^{m-2}
-(a_1-a_2)f_{n}^{m-2}f_{n+1}^{m+1}=0\,,
\label{new-red-bsqbilinear2}
\end{eqnarray}
and
\begin{eqnarray}
&&\beta H_n^{m+1}F_{n+1}^m
-\beta H_{n+1}^{m}F_{n}^{m+1}
=\alpha F_{n+1}^{m+1}F_n^m\,,\label{new-lattice-bsqbilinear0}
\\
&&\beta H_n^{m-2}F_{n+1}^m
-\beta H_{n+1}^{m}F_{n}^{m-2}
=\frac{p^3-q^3}{\alpha^2}
F_{n+1}^{m-2}F_n^m\,,\label{new-lattice-bsqbilinear1}
\\
&&\delta F_{n+1}^mF_{n}^{m-1}
+\frac{p^3-q^3}{\alpha^2}F_{n}^{m+1}F_{n+1}^{m-2}
-\alpha F_{n}^{m-2}F_{n+1}^{m+1}=0\,.
\label{new-lattice-bsqbilinear2}
\end{eqnarray}
Assuming $F_n^m=f_n^m, H_n^m=h_n^m$ and 
comparing coefficients of these equations, 
we have the relations of parameters
\[
\alpha=\left(\frac{1}{a_1}-\frac{1}{a_2}\right)\beta\,,
\quad \delta=-\frac{3\beta}{a_2}\,,\quad 
\beta=
\left[
\frac{(p^3-q^3)a_1^3a_2^3}{(2a_1+a_2)(a_1-a_2)^2}\right]^{\frac{1}{3}}\,. 
\]
\end{proof}
Note that Hietarinta and Zhang gave the Casorati determinant 
form of the $\tau$-functions $F_n^m$ and $G_n^m$\cite{Hietarinta}. 
We also note that if we introduce the auxiliary independent variable
$x$ by $\partial_x \psi_i(n,m,k;s) = \psi_i(n,m,k;s+1)$, then
$\rho_n^m$ can be expressed as $\rho_n^m = \partial_x \tau_n^m$.
\quad \\

{\bf Remark 3:}\\
In the case of $p^3-q^3=0$, Eq.($\ref{bisa2}$) is
\begin{equation}
-\alpha (p+2q)F_{n+1}^{m+1}F_n^{m-1}
-\alpha \beta G_n^{m-1}F_{n+1}^{m+1}
+\alpha \beta F_n^{m-1}G_{n+1}^{m+1}
=\gamma F_n^mF_{n+1}^m\,.
\end{equation}
Equation ($\ref{sabun}$) is written as
\begin{eqnarray}
\alpha \frac{F_n^m F_{n-1}^{m-1}}
{F_n^{m-1}F_{n-1}^{m}}
+\frac{\gamma}{\alpha}
\frac{F_{n-1}^{m+1}F_{n}^{m+1}}
{F_{n}^{m+2}F_{n-1}^{m}}
=
\alpha \frac{F_{n+1}^{m+2} F_{n}^{m+1}}
{F_{n+1}^{m+1}F_{n}^{m+2}}
+\frac{\gamma}{\alpha}
\frac{F_{n}^{m}F_{n+1}^{m}}
{F_{n+1}^{m+1}F_{n}^{m-1}}\,.
\end{eqnarray}
After some calculation, we have
\begin{equation}
\frac{\alpha F_{n-1}^{m-1}F_{n+1}^{m+1}-\frac{\gamma}{\alpha}
F_{n+1}^{m}F_{n-1}^{m}}{F_{n}^{m-1}F_n^{m+1}}
=
\frac{\alpha F_{n+1}^{m+2}F_{n-1}^{m}-\frac{\gamma}{\alpha}
F_{n-1}^{m+1}F_{n+1}^{m+1}}{F_{n}^{m}F_n^{m+2}}\,.
\end{equation}
After decoupling, we obtain
\begin{equation}
\alpha F_{n-1}^{m-1}F_{n+1}^{m+1}-\frac{\gamma}{\alpha}
F_{n+1}^{m}F_{n-1}^{m}=\delta F_n^{m-1}F_n^{m+1}\,,
\end{equation}
which is nothing but Eq.(\ref{date-bilinear}). Thus we conclude that 
this special case gives a discrete potential BSQ equation 
which has the same $\tau$-function to Eq.(\ref{date-pbsq}).

\section{Conclusion}
We have proposed 
an alternate form of
discrete potential BSQ equation.
We have constructed the bilinear equations and multisoliton solutions
for the discrete potential BSQ equation.
The bilinear equations and multisoliton solutions have been constructed by
one of 3-reductions of the Hirota-Miwa equation. 
Using the discrete potential BSQ equation, we have presented 
the ultradiscrete potential BSQ equation. 
We have also studied the lattice potential 
Boussinesq equation of Nijhoff et al.
using the singularity confinement test. 
Although the lattice potential Boussinesq 
equation of Nijhoff et al. is in very complicated form,
we can find bilinear equations easily 
by using singularity confinement test. 
We have investigated the relationships among our 
alternate discrete potential BSQ equation, 
the discrete potential BSQ equation of Date et al. 
and the lattice BSQ equation by Nijhoff et al. 

An interesting problem is to present explicit forms of 
soliton solutions of the 
ultradiscrete potential BSQ equation. 
Since 3-reduction condition is very complicated, it is not easy 
to see which solutions can survive in the ultradiscrete limit. 
We will address this problem in the near future. 

\section*{Acknowledgement(s)}

The authors would like to thank
Professor Masayuki Oikawa for stimulating discussions. 
The authors would like to acknowledge helpful 
comments and suggestions of the referees.
The authors are grateful for
the hospitality of the Isaac Newton Institute for Mathematical
Sciences (INI) in Cambridge where this article
was completed during the programme Discrete Integrable Systems (DIS).
This work was partially supported by the JSPS Grant-in-Aid for
Scientific 
Research No.21656027 and No. 19340039.

\end{document}